\documentclass[runningheads,a4paper]{llncs}
\usepackage{paralist}
\usepackage{latexsym}
\usepackage{epsfig}
\usepackage{algorithm}
\usepackage[noend]{algpseudocode}
\usepackage{amsfonts}
\usepackage{amsmath}
\usepackage{amssymb}
\usepackage{mathrsfs}
\usepackage{epsfig}
\usepackage{algorithm}
\usepackage{amsmath}
\usepackage{subfigure}
\usepackage{multirow}
\usepackage{array}
\usepackage{caption2}
\usepackage{tabularx}
\usepackage{tikz}
\usepackage{graphicx}
\usepackage{float}
\usepackage{url}
\usepackage{multicol}
\usepackage{enumerate}
\urldef{\mailsa}\path|{pengff, chm}@ios.ac.cn|
\newcommand{\keywords}[1]{\par\addvspace\baselineskip
\noindent\keywordname\enspace\ignorespaces#1}
\renewcommand{\figurename}{Figure}

\newcommand{\PreserveBackslash}[1]{\let\temp=\\#1\let\\=\temp}
\begin{document}

\mainmatter

\title{Discovering Restricted Regular Expressions with Interleaving}
\titlerunning{Discovering Restricted Regular Expressions with Interleaving}

\author{Feifei Peng$^{1,2}$
%\thanks{Please note that the LNCS Editorial assumes that all authors have used
%the western naming convention, with given names preceding surnames. This determines
%the structure of the names in the running heads and the author index.}%
 \and Haiming Chen$^1$}
%\thanks{}}
%
\authorrunning{F. Peng and H. Chen}

\institute{$^1$State Key Laboratory of Computer Science,\\
    Institute of Software, Chinese Academy of Sciences, Beijing 100190, China\\
$^2$ University of Chinese Academy of Sciences\\
\mailsa\\
}


\maketitle

\begin{abstract}
Discovering a concise schema from given XML documents is an important problem in XML applications. In this paper, we focus on the problem of learning an unordered schema from a given set of XML examples, which is actually a problem of learning a restricted regular expression with interleaving using positive example strings. Schemas with interleaving could present meaningful knowledge that cannot be disclosed by previous inference techniques. Moreover, inference of the $minimal$ schema with interleaving is challenging. The problem of finding a $minimal$ schema with interleaving is shown to be NP-hard. Therefore, we develop an approximation algorithm and a heuristic solution to tackle the problem using techniques different from known inference algorithms. We do experiments on real-world data sets to demonstrate the effectiveness of our approaches. Our heuristic algorithm is shown to produce results that are very close to optimal.
%Motivated by the fact that partial order among symbols belong to the same factor of a SURE is consistent, we are able to provide an inference algorithm based on partitioning the alphabet of example strings into a set of partial order plans. A partition algorithm $Consistent\_DAG$ which translates the strings to a corresponding directed acyclic graph and a merge algorithm $ConsistentMiner$ which merging small data sets are proposed to discover all the independent partial order plans.

\end{abstract}

\section{Introduction}
When XML is used for data-centric applications such as integration, there may be no order constraint among siblings~\cite{sibilings12}. Meanwhile, the relative order within siblings may be still important. For example, consider a ticket system with two ticket machines, where there are two bunches of tourists lining up waiting to buy tickets. Each group has two tourists. We can then define the unordered schema for the ticket system. The ordered groups preserve only the relative order of their members. This not only allows individual tourists to insert themselves within a group, but also lets two groups interleave their members. The exact XML Schema Definition (XSD) for the purchasing sequence can be essentially represented as $\allowbreak g1.m1^*g1.m2^*g2.m1^*g2.m2^*$ $|\allowbreak g2.m1^*g2.m2^*g1.m1^*g1.m2^*$ $|\allowbreak g1.m1^*g2.m1^*g1.m2^*g2.m2^*$ $
|\allowbreak g1.m1^*g2.m1^*g2.m2^*g1.m2^*$ $|\allowbreak g2.m1^*g1.m1^* \allowbreak g2.m2^*\allowbreak g1.m2^*$ $|\allowbreak g2.m1^*g1.m1^*g1.m2^*\allowbreak g2.m2^* $, where $gi.mj^*$ means the $j$th member in the $i$th group can buy zero or more tickets. It shows the length of the exact regular expression can be exponential when compared to the number of members in sequences.

Actually, $(g1.m1|g1.m2|g2.m1|g2.m2)^*$ is used in practice~\cite{bone13} instead of the minimal ones, which may permit invalid XML documents (i.e., over-permissive). For example, it may permit the second member in the sequence of the first group to purchase tickets before the first member. There are many negative consequences of over-permissive~\cite{bone13}. Thus it is necessary to study how to infer an unordered minimal schema for this kind of XML documents.
\raggedbottom

Previous researches on XML Schema inference have been done mainly in the context of ordered XML, which can be reduced to learn regular expressions. Gold~\cite{gold67} showed the class of regular expressions is not identifiable in the limit. Therefore numerous papers~(e.g.\cite{BEX09,inferschema,{inferdeter},fast13}) studied inference algorithms of restricted classes of regular expressions. Most of them were based on properties of automata. Bex et al.~\cite{BEX09} proposed learning algorithms for single occurrence regular expressions (SOREs) and chain regular expressions (CHAREs). Freydenberger and K\"{o}tzing~\cite{fast13} gave more efficient algorithms learning a minimal generalization for the above classes. The approach is based on descriptive generalization~\cite{fast13} which is a natural extension of Gold-style learning.

However, there is no such kind of automata for regular expressions with interleaving since they do not preserve the total order among symbols. Thus we have to explore new techniques. While Ciucanu~\cite{Ciu13} proposed learning algorithms for two unordered schema formalisms: disjunctive multiplicity schemas (DMS) and its restriction, disjunction-free multiplicity schemas (MS), both of them disallow concatenation within siblings. Thus they are less expressive than ours. Moreover, the ordering information in our schema formalism can not be fully captured by the three characterizing triples used to construct a DMS or MS.

Inference algorithms in this paper use some similar techniques with algorithms mining global partial orders from sequence data~\cite{frag03,glob00,frequent06}. However, the semantic concepts there are typically quite different from ours. Mannila et al.~\cite{glob00} tried to find mixture models of parallel partial orders. However, to learn unordered regular expressions, series parallel orders may not be sufficient since they can conflict with some data in the whole data set. Another restriction in the above method is that it can only be applied to strings where each symbol occurs at most once. Particularly, Gionis et al.~\cite{frag03} emphasised on recovering the underlying ordering of the attributes in high-dimensional collections of 0-1 data. An implicit assumption is that attribute can also occur at most once. For learning regular expressions with interleaving, symbols in strings can present any times and partial orders among siblings are independent with no violations. Hence many techniques from data mining are not directly applicable. Therefore, learning restricted regular expressions with interleaving remains a challenging problem.

In this paper, we address the problem of discovering a minimal regular expression with interleaving from positive examples. The main contributions of the paper are listed as follows:
\begin{enumerate}[-]
\item We propose a better and more suitable formalism to specify precise unordered XML: the subset of regular expressions with interleaving (SIREs). SIREs can express the content models succinctly and concisely. For example, the above example can be depicted as $(g1.m1^*g1.m2^*)\&(g2.m1^*g2.m2^*)$.

\item We introduce the notion of SIRE-minimal in the terminology of~\cite{fast13} and some properties of SIRE-minimal.
\item We prove the problem of finding a minimal SIRE is NP-hard and develop an approximation algorithm conMiner to find solutions with worst-case quality guarantees and a heuristic algorithm conDAG that mostly finds solutions of better quality as compared to the approximation algorithm conMiner.

\item We conduct experiments comparing our methods with Trang~\cite{Trang} on real world data, incorporating small and large data sets. Our experiments show that conMiner and conDAG outperform existing systems on such data.
\end{enumerate}
The rest of the paper is organized as follows. Section~\ref{sect:pre} contains basic definitions. In Section~\ref{sect:des} we discuss properties of minimal-SIRE. In Section~\ref{sect:infer} an approximation algorithm conMiner and a heuristic algorithm conDAG are proposed. Section~\ref{sect:expre} gives the empirical results. Conclusions are drawn in Section~\ref{sect:conwk}.

\section{Preliminaries}
\label{sect:pre}
%This section introduces formal interleaving semantics, the class of regular expressions and the notion of consistent partial order that used in this paper.
Let $u$ and $v$ be two arbitrary strings. By $u\&v$ we denote the set of strings that is obtained by interleaving of $u$ and $v$ in every possible way. That is, $u\&\varepsilon=\varepsilon\&u=u$, $v\&\varepsilon=\varepsilon\&v=v$. If both $u$ and $v$ are non-empty let $u=au',v=bv'$, $a$ and $b$ are single symbols, then $u\&v=a(u'\&v)\cup b(u\&v')$. Let $\Sigma$ be an alphabet of symbols. The regular expressions with interleaving over $\Sigma$ are
defined as: $\emptyset, \varepsilon$ or $a\in \Sigma$ is a regular expression, $E_1^?$, $E_1^*$, $E_1^+$, $E_1E_2$, $E_1|E_2$, or $E_1\&E_2$ is a regular expression for regular expressions $E_1$ and $E_2$. They are denoted as RE(\&). The language described by $E$ is defined as follows: $L(\emptyset)=\{\emptyset\};$ $L(\varepsilon)=\{\varepsilon\};$ $L(a)=\{a\};$
$L(E_1^?)=L(E_1)^?;$ $ L(E_1^+)=L(E_1)^+;$ $ L(E_1^*)=L(E_1)^*;$
$L(E_1E_2)=L(E_1)L(E_2);$ $L(E_1|E_2)=L(E_1)\cup L(E_2);$  $L(E_1\&E_2)=L(E_1)\&L(E_2)$. 
We consider the subset of regular expressions with interleaving (SIREs) defined by the following grammar.
\begin{definition}The restricted class of regular expressions with interleaving (RREs) are $RE(\&)$ over $\Sigma$ by the following grammar for any $a\in \Sigma$:
\begin{equation*}
\begin{array}{rcl}
S::&=&T\&S|T\\
T::&=&\varepsilon|a|a^+|a^?|a^*|TT\\
\end{array}
\end{equation*}
\end{definition}
The subset of regular expressions with interleaving (SIREs) are those RREs in which every symbol can occur at most once. Since SIREs disallow repetitions of symbols, they are certainly deterministic and satisfy the UPA constraint required by the XML specification.

A partial order $M$ for a string $s$ is a binary relation that is reflexive, antisymmetric and transitive. We write $a \prec b$ if $a$ is before $b$ in the partial order. For string $s=x_1\cdots x_l$, the \textit{transitive closure} of $s$ is denoted by $tr(s)=\{(x_i,x_j)|1\leq i<j\leq l\}$, where $l$ is the length of $s$. For example $s=abcd$, $tr(s)=\{ab,ac,ad,bc,bd,cd\}$.

A partial-order set $t$ is a set of symbols together with a partial ordering. We say $ab\in t$ if $a$ precedes $b$ in every string in a string collection. Consistent partial order set (CPOS) $T$ is a set which contains all the disjoint partial-order sets $t_i$ of the given examples. For example, consider $W=\{abcd,dabc\}$. Obviously, $a\prec b\prec c$, $T=\{abc,d\}$. The connection between CPOS and SIRE is directly. That is, given a CPOS, we can write it to the form of SIRE by combining all the elements in CPOS with $\&$. For example, in this case the corresponding SIRE $s=abc\&d$. Therefore, the problem of finding a minimal SIRE can be reduced to the problem of finding a minimal CPOS.
\section{Descriptivity}
\label{sect:des}
This section introduces the notion of minimal expressions. Roughly speaking minimal is the greatest lower bound of a language $L$ within a class of expressions, which is conceptually similar with $infimum$ in the terminology of mathematics.
\begin{definition}[\cite{fast13}]Let $\mathcal{D}$ be a class of regular expressions over some alphabet $\Sigma$. A $\delta\in\mathcal{D}$ is called $\mathcal{D}$-minimal of non-empty language $S\subseteq\Sigma^*$ if $L(\delta)\supseteq S$, and there is no $\gamma\in\mathcal{D}$ such that $L(\delta)\supset L(\gamma)\supseteq S$.
\end{definition}
\begin{proposition}
\label{pro:num}
Let $n$ be the number of alphabet symbols. The number of pairwise non-equivalent $SIREs$ is $\mathcal{O}(n!)$.
\end{proposition}
\begin{proof}
Disregarding operators ?,+,*, the number of SIREs over a finite $\Sigma$ is equivalent to the number of ordered partitioning $|\Sigma|$ symbols. The number of these partitions is given by the $|\Sigma|$th ordered Bell numbers~\cite{de09}. For instance, if $\Sigma=\{a,b,c\}$, the $3$th ordered Bell number $a(3)=13$, and the ordered partitions of $\{a,b,c\}$ is $\{abc, acb, bac, bca, cab, cba, ab\&c, ba\&c, ac\&b, ca\&b, bc\&a, cb\&a, a\&b\&c\}$. They are also distinct partitions of SIREs over $\Sigma$. The ordered Bell number~\cite{bail98} can be approximated as
$a(n)=\sum_{k=0}^{n}k!\binom{n}{k}\approx\frac{n!}{2(ln2)^{n+1}}$.
Since every symbol $a$ in $\Sigma$ has four forms which can be represented as $a, a^?, a^+~and~a^*$, the number of SIREs over $\Sigma$ is $4^na(n)$. Then $s(n)\approx\frac{4^nn!}{2(ln2)^{n+1}}$.
\qed
\end{proof}
We can then prove the existence of minimal regular expressions for SIRE.
\begin{proposition}
\label{pro:everyhas}
Let $\Sigma$ be a finite alphabet. For every language $L\subseteq\Sigma^*$, there exists a SIRE-minimal SIRE $\delta_s$.
\end{proposition}
%The proof can make a direct use of the proof of Proposition $15$ from~\cite{fast13} with Proposition~\ref{pro:num}, thus was omitted for space reasons.
\begin{proof}Assume there is a language $L$ over $\Sigma$ such that no expression $\alpha\in SIRE$ is SIRE-minimal. This implies that there is an infinite sequence $(\beta_i)_{i\geq0}$ of expressions
from SIRE with $\alpha=\beta_0$ and $L(\beta_i)\supset L(\beta_{i+1})\supseteq L$ for all $i\geq0$. This contradicts the
fact that there are only a finite number of non-equivalent SIREs over $\Sigma$ by Proposition~\ref{pro:num}.
\qed
\end{proof}

\begin{proposition}
\label{pro:minimal}
For any example string set $E$ over $\{a_1,\cdots,a_n\}$, let $S=s_1\&\cdots\allowbreak\&\allowbreak s_l$ be a SIRE such that $E\subseteq L(S)$. $S$ is a minimal SIRE if and only if:\\
(1) the number of $s_i$ is minimized and\\
(2) the size of each $s_i$ is as large as possible.
\end{proposition}
The proof was omitted for space reasons.

In other words, a minimal SIRE is the most specific SIRE that consistent with the given example strings. For instance, all of $S_1=a\&bc\&d$, $S_2=abc\&d$ and $S_3=ad\&bc$ can accept $E=\{abcd,adbc\}$. However, since $S_1=(ad|da)\&bc=(ad\&bc)|(da\&bc)=S_3|(ad\&bc)$, we can get $L(S_1)\supset L(S_3)$ which means $S_1$ is not minimal. As for $S_2$ and $S_3$, since $L(S_2)=\{abcd,abdc,adbc,dabc\}$ and $L(S_3)=\{bcad,bacd,badc,abcd,abdc,adbc\}$, this means $S_3$ is not minimal. As we shall see, $S_2$ is a better approximation of E. In fact, $S_2$ can be verified to be a minimal by referring to Proposition~\ref{pro:minimal}.

\section{Minimal SIREs}
\label{sect:infer}
In this section, we first prove finding a minimal SIRE for a given set of strings is NP-hard by reducing from finding a maximum independent set of a graph, which is a well-known NP-hard graph problem~\cite{MIS14}. Then we present learning algorithms that construct approximatively minimal SIREs.
\subsection{Exact Identification}
First, we introduce the notion of maximum independent set of a graph~\cite{MIS14}. Consider an undirected graph $G(V,E)$, an independent set (IS) is a set $I\subseteq V$ such that $\forall{u,v\in V},~(u,v)\notin E$. The maximum independent set (MIS) problem consists in computing an IS of the largest size. Next, we define the problem \verb"all_mis" which takes a graph $G$ as input, finding a MIS $S'$ of G by applying function \verb"max_independent_set", and repeating the step for subgraph $G[V-S']$ until there exists no vertex in the subgraph. In other words, \verb"all_mis" is to divide $V$ into disjoint subsets by \verb"max_independent_set". Clearly, problem \verb"all_mis" is NP-hard.

The main idea of finding a minimal SIRE is based on the observation that there are sets of conflicting siblings that cannot be divided into the same subset of CPOS. A pair $xy$ is called forbid pair in a string database if both $xy$ and $yx$ exists in the transitive closure of strings. The set of forbid pairs is called a $constraint$. By Proposition~\ref{pro:minimal}, if we split the set of symbols in a $constraint$ into several subsets $t_1,\cdots,t_n$ such that $n$ is minimized and for each $i\in[1..n]$, $t_i$ is the longest of its alternatives. Then the set of $t_i$ where $i\in[1..n]$, is a minimal CPOS which can be transformed to a minimal SIRE.
\begin{lemma}Minimal SIRE finding problem is NP-hard.
\end{lemma}
\begin{proof}
We demonstrate that \verb"all_mis" can be reduced in polynomial time to minimal SIRE finding problem. Given an instance of \verb"all_mis", we can generate a corresponding instance of minimal SIRE finding as follows. For the graph $G$ in \verb"all_mis", the reduction algorithm computes the $constraint$ set by adding all edges in $G$ to $constraint$, which is easily obtained in polynomial time. The output of the reduction algorithm is the instance set $constraint$ of minimal SIRE finding problem. $t_i$ in CPOS is the longest of its alternatives if and only if \verb"all_mis" computes a maximum independent set at the $i$th step. Thus, minimal SIRE finding problem is equivalent to the original \verb"all_mis". Since \verb"all_mis" is NP-hard, minimal SIRE finding problem is NP-hard.
\qed
\end{proof}
\subsection{Approximation Algorithm}
The process of this approach is formalized in Algorithm~\ref{alg:consi}. Algorithm~\ref{alg:consi} works in four steps and we illustrate them on the sample $E=\{abcd,aadbc,bdd\}$. The first step (lines 1-2) computes the non-constraint and constraint set using the function \verb"tran_reduction". The transitive~closure of $E$ is $tr=\{ab,ac,ad,bc,bd,db,\allowbreak dc\}$. Add $uv$ to $constraint$ if $vu\in tr$. Add $uv$ to $L_2$ otherwise. We get $L_2=\{ab,ac,ad,bc\}$ and $constraint=\{bd,cd,db,dc\}$. Construct an undirected graph $G$ using element in $constraint$ as edges. The second step (lines 3-7) is to select a MIS of $G$, add it to list $allmis$ and delete the MIS and their related edges from $G$. The process is repeated until there exists no nodes in $G$. The problem of finding a maximum independent set is an NP-hard optimization problem. As such, it is unlikely that there exists an efficient algorithm for finding a maximum independent set of a graph. However, we can find a MIS in polynomial time with a approximation algorithm, e.g. the \verb"clique_removal" algorithm proposed in~\cite{ramsey92} that finds the approximation of maximum independent set with performance guarantee $\mathcal{O}(n/(\log n)^2)$ by excluding subgraphs. For graph $G$, we obtain $allmis=\{\{b,c\},d\}$. Next, we add the non-constraint symbols to the first MIS. Then we have $allmis=\{\{a,b,c\},d\}$. The third step (lines 8-10) computes the topological sort for all subgraphs induced by subset of $L_2$ and add the result to $T$. For the sample, it returns $T=\{abc,d\}$. Finally, the algorithm returns the SIRE whose corresponding counting operators $1,*,+,?$ can be inferred using technique in algorithm \verb"CRX"~\cite{BNS06}. For the sample, it returns $a^*bc?\&d^+$.
\begin{algorithm}{\footnotesize}
\caption{{\it conMiner}($W$)}
\label{alg:consi}
\begin{algorithmic}[1]
\Require Set of words $W=\{w_1,...,w_n\}$
\Ensure a minimal SIRE $T$
\State $L_2,constraint=tran\_reduction(W,T)$
\State $G=Graph(constraint)$
\While {$G.nodes()!=null$}
\State $v=clique\_removal(G)$
\State $G=G-v$
\State $allmis.append(v)$
\EndWhile
\State $allmis[0]=allmis[0].union(alphabet(L2)-alphabet(constraint))$
\For {each $mis\in allmis$}
\State $H=Graph(mis,L2)$
\State $T.append(topological\_sort(H))$
\EndFor
\State\Return $learner_{oper}(W,T$)
\end{algorithmic}
\end{algorithm}

\subsection{Heuristic Algorithm}
Although a number of approximation algorithms and heuristic algorithms have been developed for the maximum independent set problem, on any given instance, they may produce a SIRE that is very far from optimal. We introduce a heuristic directed acyclic graph construction algorithm directly computing a minimal SIRE. The main idea is to cluster the vertices of the existing directed graph into several disconnected subgraphs. The graph is constructed incrementally to preserve CPOS within each vertex using a greedy approach. The pseudocode of algorithm \verb"conDAG" is given in Algorithm~\ref{alg:consis}.

The input to this algorithm is the same as the input of the \verb"conMiner". The algorithm maintains lists $p,q$ as records to keep track of pairs violating the partial order constraint and lists $s,t$ to record pairs violating the partial order constraint of the string under reading. Note that $(a,b)$ violating the partial order constraint means there exist some $w_1,w_2\in W$ such that $a\prec b$ in $w_1$ and $b\prec a$ in $w_2$.

Let $ab$ be two adjacent symbols in a word $w$. The \verb"add_or_break" function checks whether edge $ab$ is added to the present graph G. If there exists no path from $b$ to $a$, no path from $a$ to $b$ in G and edge $ab$ will not make a connection between some $p[i]$ and $q[i]$, we add edge $a\rightarrow b$ in G. Self-loops such as $f\rightarrow f$ are always ignored since they have no influence on the partial order constraints. However, if there exist paths from $b$ to $a$ in $G$, $(a,b)\notin(p[i],q[i]),(q[i],p[i])$ and $a,b$ are not in $p[i],q[i]$ at the same time for all $i<len(p)$, we should break all paths from $b$ to $a$. The breakpoint can be found as below. Suppose there exists a path $u=b\alpha_1...a$, $\alpha_{0}=b$ in $G$, and substring of $w$ over $\{b,\alpha_1,...,a\}$ is $\alpha_i...a$, then we delete edge $\alpha_{i-1}\rightarrow \alpha_i$, add edge $\beta\rightarrow \alpha_i$ for all nodes $\beta$ that $\beta\rightarrow b$, and add edge $\alpha_{i-1}\rightarrow \gamma$ for all nodes $\gamma$ that $a\rightarrow \gamma$. In the end, add $b\alpha_1...\alpha_{i-1}$ to $p$,$s$ and add $\alpha_i...a$ to $q$,$t$.
\begin{figure}[h]
\centering
\includegraphics[width=0.5\textwidth]{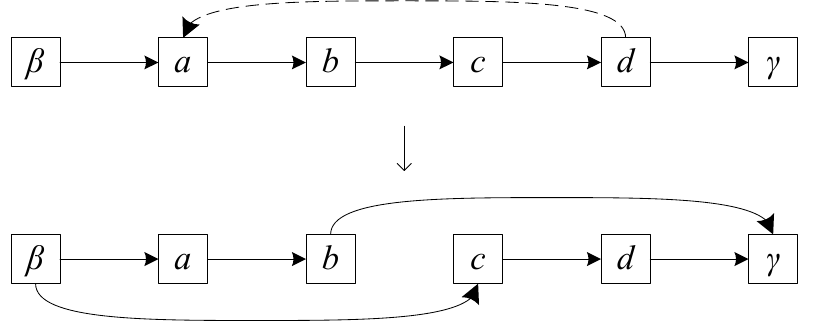}
\caption{This is an example to find the breakpoint}
\label{fig:addorbreak.pdf}
\end{figure}

Example in Figure~\ref{fig:addorbreak.pdf} shows how the function works. $W=\{\beta abcd\gamma,cda\}$, initialize empty list $p$,$q$,$s$,$t$ and empty graph $G$. After reading $w_1$, list $p$,$q$,$s$,$t$ are still empty. When reading $da\in w_2$, there already exists a path $abcd$ and $(d,a)\notin (p[i],q[i]),(q[i],p[i])$. We should break $abcd$. Since $substring(w_2,\{a,b,c,d\})=cda$, breakpoint is $c$. Then we delete edge $b\rightarrow c$, and add edges $\beta\rightarrow c$,$b\rightarrow \gamma$. In the end, add $ab$ to $p$,$s$ and add $cd$ to $q$,$t$.

The \verb"consistent" function scans the whole string $w$ by sequence to execute \verb"add_or_break" function. Each time after reading two adjacent symbols $ab$, for all pairs $(\alpha_1a\alpha_2,\alpha_3 c)$ or $(\alpha_3 c,\alpha_1a\alpha_2)\in (s,t)$, handle $cb$ likewise. Because $(\alpha_1a\alpha_2,\alpha_3 c)$ or $(\alpha_3 c,\alpha_1a\alpha_2)\in (s,t)$ declare $a\prec c$ and $c\prec a$ are in $w$, if $a\prec b$ in $w$, $c\prec b$ is also in $w$. Consider $acab$ as an example, $c$ and $a$ have been two parts after reading $ca$, $a$ has been added to $p$ and $s$ and $c$ added to $q$ and $t$. After reading the next two symbols $ab$, add edge $a\rightarrow b$. Next we should consider $cb$ since $a\in s[0],c\in s[0]$, thus add edge $c\rightarrow b$. The \verb"topological_sort(g)" construct a topological ordering of DAG in linear time. The \verb"learner_oper" is used to infer operators $?,+,*$ for each vertex.

The \verb"conDAG" algorithm combines all the functions. The constructed
graph is denoted by $G$ and the corresponding set of partitions by $C$. In each iteration, it invokes \verb"consistent" to update $G$ using the $i$th string. Then it adds all the paths from the set of vertices of in-degree zero to the set of vertices of out-degree zero. To be able to calculate the largest independent partial-order plans, a preprocessing phase is implemented. First, we consider the elements of $C$ in decreasing order of size. In each iteration, whenever we find two elements that the one contains elements of $p[i]$ and the other one contains elements of $q[i]$, we updates the shorter one by removing the common elements. Next, we merge all the lists in $C$ that share common elements. The preprocess terminates when every symbol is included in one and only one list. The following steps of the algorithm are the same as the third and the forth step of the \verb"conMiner".
\begin{algorithm}[t!]
\begin{algorithmic}[1]
\Function{$consistent$}{$G,w,p,q$}
\State $s,t:=\emptyset$, $i:=1$
\While {$i<|w|-1$}
\If {$w[i]\neq w[i+1]\wedge(w[i],w[i+1])\notin (p,q),(q,p)$}
\State $add\_or\_break(G,w,w[i],w[i+1],p,q,s,t)$
\EndIf
\For{$j:=1$ to $|s|$}
\If {$(w[i]\in s[j])\wedge((t[j][-1],w[i+1])\notin(p,q))$}
\State $add\_or\_break(G,w,last\_symbol(t[j]),w[i+1],p,q,s,t)$
\EndIf
\If {$(w[i]\in t[j])\wedge(s[j][-1],w[i+1])\notin(p,q))$}
\State $add\_or\_break(G,w,last\_symbol(s[j]),w[i+1],p,q,s,t)$
\EndIf
\EndFor
\State $i++$
\EndWhile
\EndFunction
\end{algorithmic}
\end{algorithm}

\begin{algorithm}[htbp]
\caption{{\it $conDAG$}($W$)}
\label{alg:consis}
\begin{algorithmic}[1]
\Require Set of unordered words $W=\{w_1,...,w_n\}$
\Ensure a minimal SIRE
\State $L_2,constraint=tran\_reduction(W,T)$
\State initialize graph G, $p,q:=\emptyset$
\For {$i:=1$ to n}
\State $consistent(G,w_i,p,q)$
\EndFor
\State $C=all\_paths(G,source,destination)$
\State remove the common elements from the shorter of $c_i,c_j\in C$ if $c_i[m]+c_j[n]\in constraint$.
\State merge all lists that share common elements in $C$
\For {each $mis$ in $C$}
\State $H=Graph(mis,L_2)$
\State $T.append(topological\_sort(H))$
\EndFor
\State\Return $learner_{oper}(W,T$)
\end{algorithmic}
\end{algorithm}

The time complexity analysis of this algorithm is straightforward. $add\_or\_brea\allowbreak k\allowbreak(G,\allowbreak w,a,\allowbreak b,p,q,s,t)$ can find all possible paths between two given nodes by modifying the DFS which needs $\mathcal{O}(|V| + |E|)$ steps. Breaking a circle requires $\mathcal{O}(|V|)$. Therefore, an overall time complexity for $add\_or\_break$ is $\mathcal{O}(c|V| + |E|)$, where $c$ is number of paths between the given nodes in the graph. When there exist $n(n-2)/2$ inconsistent terms in W, every two symbols are not in a group, which is the worst case. When tackling of $\alpha_{i-1}\alpha_i$, $len(p)=(n-i+1)(n-i)/2$, deciding whether $(\alpha_{i-1},\alpha_i)\in (p[j],q[j]),(q[j],p[j])$ needs $(n-i+1)(n-i)$ time. Deciding whether $\alpha_i\in s[j],t[j]$ needs $n-i$ time. There is only one path between two nodes, thus $c=1$. So the total time of $consisitent$ is $\sum_{i=2}^{n}(n-i)^2(|V|+|E|)$ where $|V|=n$, and $|E|$ is $\mathcal{O}(n)$ according to the analysis above.

The $tran\_reduction$ computation requires $\mathcal{O}(n^2)$ time, where $n$ is the number of distinct symbols. Each iteration requires $\mathcal{O}(n^3)$ time to maintain the graph. Computing all paths from source to destination can be done in $\mathcal{O}(n^2)$ time, and $topological\_sort(g)$ constructs a topological ordering of DAG in linear time, thus $\mathcal{O}(|V| + |E|)$ steps are sufficient. Inference of operators $?,+,*$ needs time $\mathcal{O}(m)$. Hence the time complexity of the algorithm is $\mathcal{O}(tn^4+m)$, where $m$ is the sum of length of the input example strings, $n$ the number of alphabet symbols and $t$ the number of strings.

To illustrate our algorithm, consider the example $E=\{abcd,aadbc,bdd\}$, $L_2=\{ab,ac,ad,bc\}$, $constraint=\{bd,cd,db,dc\}$ in the above section. A directed graph which consists of vertex $V=\{a,b,c,d\}$ and edges $E=\{ab,bc,ad\}$ can be obtained. $p=\{bc\}$ and $q=\{d\}$. All paths from source to destination are $C=\{abc,ad\}$. Since $bd\in constraint$, $C[2]$ is updated by removing the common elements between $C[1]$ and $C[2]$. $C[2]$ is $d$. The final $C$ is $\{abc,d\}$. The following steps are the same.
\section{Experiments}
\label{sect:expre}
In this section, we validate our approaches on real-life DTDs, and compare them with that of Trang~\cite{Trang}. All experiments were conducted on an IBM T400 laptop computer with a Intel Core 2 Duo CPU(2.4GHz) and 2G memory. All codes were written in python.

The number of corpora of XML documents with an interesting schema is rather limited. We obtained our real-life DTDs from the XML DATA repository maintained by Miklau~\cite{schema}. Unfortunately, most of them are either not data-centric or not with a DTD. Specifically, We chose the DBLP Computer Science Bibliography corpus, a data-centric database of information on major computer science journals and proceedings.
\begin{table}[H]
\begin{center}
\tabcolsep 2.5pt
\begin{tabular}{|l|m{10.25cm}<{\centering}|} \hline
\multirow{8}{*} \textbf{Element} & Original DTD\\
name& Exact Minimal DTD\\
Sample& Result of conMiner\\
size& Result of conDAG \\
&Result of Trang\\
Number of&Simplified Exact Minimal DTD\\
interleaving&Simplified Result of conMiner\\
&Simplified Result of conDAG\\\hline
\multirow{3}{*}\textbf{inproceedings}&$(a_1|a_2|\cdots|a_{22})^*$\\
$2122274$&${a_1}^*a_{12}?{a_5}^*a_9?a_{18}?{a_{15}}^*\&a_3a_6{a_{11}}^*
\&{a_{19}}^*\&{a_{13}}^*\&a_4\&{a_{14}}^*$\\
$2122274$&${a_5}^*a_{18}?{a_{15}}^*\&a_{12}?a_9?
{a_{13}}^*\&{a_1}^*{a_{14}}^*\&a_6{a_{11}}^*\&a_3\&a_4\&{a_{19}}^*$\\
$2122274$&${a_1}^*a_4a_9?{a_{11}}^*{a_{15}}^*
\&a_3a_{12}?{a_5}^*a_{18}?\&{a_{13}}^*\&a_6\&{a_{14}}^*\&{a_{19}}^*$\\
$2122274$&$(a_1|a_3|a_5|a_6|a_9|
a_{11}|a_{12}|a_{13}|a_{14}|a_{15}|a_{18}|a_{19})^+$\\

5&6\&3\&1\&1\&1\&1\\
6& 3\&3\&2\&2\&1\&1\&1\\
5& 5\&4\&1\&1\&1\&1\\
\hline
\multirow{3}{*}\textbf{article}&$(a_1|a_2|\cdots|a_{22})^*$\\
111608&${a_1}^*a_{17}?a_5^*a_{12}?{a_{15}}^*\&a_3a_6a_{11}?\&{a_{13}}^*
\&a_8\&a_{10}?\&{a_{14}}^*\&a_9?$\\
111608&$a_{17}?a_{12}?a_9?{a_{15}}^*\&{a_1}^*a_6a_{11}?\&a_3\&
a_5^*\&{a_{13}}^*\&a_8\&a_{10}?\&{a_{14}}^*$\\
111608&${a_3}^*a_{17}?a_6a_{11}?\&{a_1}^*a_8a_{12}?{a_{15}}^*\&{a_{13}}^*
\&a_5^*\&a_{10}?\&a_{12}?\&a_9?$\\
111608&$a_2?(a_1|a_3|a_5|a_6|a_8|a_9|a_{10}|a_{11}|a_{12}|a_{13}
|a_{14}|a_{15}|a_{17})^+$\\

6& 5\&3\&1\&1\&1\&1\&1\\
7& 4\&3\&1\&1\&1\&1\&1\&1\\
6 &4\&4\&1\&1\&1\&1\&1\\

\hline
\multirow{3}{*}\textbf{proceedings}&$(a_1|a_2|\cdots|a_{22})^*$\\
$3007$&${a_2}^*{a_3}^+a_{18}?a_{21}?a_{8}?a_{10}?a_{13}?a_{12}?{a_{15}}^*a_{19}?
a_7?a_9?\&a_4?\&a_{17}?\&a_6\&a_{20}^*\&a_{11}?$\\
$3007$&${a_2}^*{a_3}^+a_{19}?a_{13}?a_{20}^*{a_{15}}^*a_{12}?\&a_4?
a_7?a_{8}?a_9?\&a_{21}?a_{18}?a_{10}?\&a_6\&a_{17}?\&a_{11}?$\\
$3007$&${a_2}^*{a_3}^+a_{8}?a_{18}?a_{21}?a_{10}?a_9?a_{19}?a_{13}?a_7?{a_{15}}^*
\&a_4?a_{12}?\&a_{17}?\&a_6\&a_{20}^*\&a_{11}?$\\
$3007$&${a_2}^*{a_3}^+(a_4|a_6|a_7|a_8|a_9|a_{10}|
a_{11}|a_{12}|a_{13}|a_{17}|a_{18}|a_{19}|a_{20}|a_{21})^+{a_{15}}^*$\\

5&12\&1\&1\&1\&1\&1\\
5&7\&4\&3\&1\&1\&1\\
5&11\&2\&1\&1\&1\&1\\

\hline
\multirow{3}{*}\textbf{incollection}&$(a_1|a_2|\cdots|a_{22})^*$\\
$1009$&${a_1}^*a_3a_4a_{17}?a_{20}?a_{16}?a_{11}?{a_{15}}^*a_{14}?
\&a_{13}?a_{19}?\&a_{5}?\&a_{6}$\\
$1009$&${a_1}^*a_3a_{17}?a_{6}\&{a_{15}}^*a_{13}?a_{16}?a_{14}?\&
a_4a_{11}?\&a_{20}?a_{19}?\&a_{5}?$\\
$1009$&${a_1}^*a_3a_4a_{17}?a_{11}?{a_{15}}^*a_{14}?\&a_{6}a_{20}?
\&a_{5}?a_{16}?\&a_{13}?\&a_{19}?$\\
$1009$&$(a_1|a_3|a_4|a_5|a_6|a_{11}|a_{13}|a_{16}|a_{17}
|a_{20})^+(a_{14}|{a_{15}}^*)$\\

3&9\&2\&1\&1\\
4&4\&4\&2\&2\&1\\
4&7\&2\&2\&1\&1\\

\hline
\multirow{3}{*}\textbf{phdthesis}&$(a_1|a_2|\cdots|a_{22})^*$\\
$72$&$a_1a_3a_6a_{17}?a_{21}?a_{20}?a_{9}?a_{13}?a_{12}?\&a_{22}$\\
$72$&$a_1a_3a_6a_{12}?a_{21}?a_{22}a_{13}?a_{20}?\&a_{17}?a_{9}?$\\
$72$&$a_1a_3a_6a_{17}?a_{21}?a_{20}?a_{13}?a_{9}?a_{12}?\&a_{22}$\\
$72$&$a_1a_3a_6(a_{12}|a_{21})?(a_9|a_{17}|a_{22})^+(a_{13}|a_{20})?$\\

1&9\&1\\
1&8\&2\\
1&9\&1\\

\hline

\multirow{3}{*}\textbf{www}&$(a_1|a_2|\cdots|a_{22})^*$\\
$38$&${a_1}^*{a_2}^*a_3a_4?a_6?a_{11}$\\
$38$&${a_1}^*{a_2}^*a_3a_4?a_6?a_{11}$\\

$38$&${a_1}^*{a_2}^*a_3a_4?a_6?a_{11}$\\
$38$&$({a_1}^*|{a_2}^*)a_3a_4?a_6?a_{11}$\\

0&6\\
0&6\\
0&6\\

\hline

\end{tabular}
\end{center}
\vskip 10pt
  \caption{Results of exact algorithm, conMiner, conDAG and Trang on DTDs}
\label{tab:schema}
\end{table}

Table~\ref{tab:schema} lists the non-trivial element definitions in the above mentioned DTD together with the results derived by exact algorithm, heuristic algorithm conMiner, approximation algorithm conDAG, and Trang. We implement 
the exact algorithm using conMiner by replacing function \verb"clique_removal" with an exponential time algorithm proposed by S. Tsukiyama~\cite{allmis92}. We also list the number of interleavings used and the simplified of our results to have a clear view of their relationship. The numbers in the first column the first five rows in each element refer to the element name and the sample size respectively. The numbers in the first column the last three rows in each element refer to the number of interleavings used by the result of exact algorithm, conMiner and conDAG, respectively. It can be verified that all expressions learned by exact algorithm, conDAG and conMiner are more strict than that of Trang and the original DTDs which indicates there exists much more over-permissive in both the original DTDs and the results of Trang.

We note that there may exist many minimal expressions given a set of unordered strings. For instance, for \verb"phdthesis", the form of the result of conDAG is the same with the exact minimal expression. The orders among symbols of their first siblings, however, differ widely. This is due to the fact that a diagraph may have several different topological sorts. Therefore, we ignore the sequel in the symbols and only compare their simplified form. The table shows clearly that conDAG yields concise super-approximations to the exact minimal expressions. Although for \verb"proceedings, incollection and phdthesis", the expressions produced by conMiner and conDAG have the same number of interleavings, conDAG yields longer length of siblings and thus finds solutions of better quality as compared to the solutions found by the approximation algorithm.
%\subsection{Synthetic target expressions}
\section{Conclusion}
\label{sect:conwk}
This paper proposes a strategy for learning a class of regular expressions with interleaving: first, compute consistent partial order $T$, then equip each factor with counting operators. As future work, we will investigate several interesting problems inspired by this study. First, we would like to extend our algorithms for more expressive schemas, for example schemas allow disjunction $``|"$ within siblings. Second, how to extend algorithms to mine all independent frequent closed partial orders~\cite{frequent06} is also an attractive topic.

\section*{Acknowledgement}We thank the users of Stack Overflow~\cite{stackoverflow}, for reminding us the maximum independent set problem.

\end{document}